\documentclass[11pt]{article}

\usepackage{amssymb,amsmath}
\usepackage{amsthm}
\usepackage{bbm} 
\usepackage{amsfonts,mathrsfs}
\usepackage{url}
\usepackage{enumerate}
\usepackage{txfonts}

\def\ot{\otimes}
\def\H{\textsf{H}}\def\S{\textsf{S}}\def\T{\textsf{T}}
\newcommand{\inner}[2]{\langle #1 , #2\rangle}
\newcommand{\out}[2]{| #1\rangle\langle #2 |}

\newcommand{\trans}{{\scriptstyle\mathsf{T}}}

\newcommand{\pa}[1]{(#1)}
\newcommand{\Pa}[1]{\left(#1\right)}

\newcommand{\set}[1]{\{#1\}}

\newcommand{\ket}[1]{|#1\rangle}

\def\Jamiolkowski{J}
\newcommand{\jam}[1]{\Jamiolkowski\pa{#1}}

\DeclareMathOperator{\trace}{Tr}
\newcommand{\ptr}[2]{\trace_{#1}\pa{#2}}

\newcommand{\tr}[1]{\ptr{}{#1}}

\newcommand{\fontmapset}{\mathbf} 

\newcommand{\mset}[2]{\fontmapset{#1}\pa{#2}}

\newcommand{\lin}[1]{\mset{L}{#1}}

\newcommand{\identity}{\mathbbm{1}}
\newcommand{\idsup}[1]{\identity_{#1}}

\newcommand{\Natural}{\mathbb{N}}
\newcommand{\Real}{\mathbb{R}}
\newcommand{\Complex}{\mathbb{C}}

\def\cH{\mathcal{H}}

\def\rM{\mathrm{M}}

\newtheorem{thrm}{Theorem}[section]
\newtheorem{lem}[thrm]{Lemma}

\newtheorem{cor}[thrm]{Corollary}
\theoremstyle{definition}

\newtheorem{remark}[thrm]{Remark}

\numberwithin{equation}{section}

\newcommand{\Innerm}[3]{\langle #1 | #2 | #3 \rangle}

\author{Lin Zhang and Junde Wu}

\date{\small{\it Department of Mathematics, Zhejiang University, Hangzhou, 310027, P.~R.~China\\E-mail: linyz@zju.edu.cn, wjunde@126.com, wjd@zju.edu.cn}}

\begin{document}


\title{Von Neumann Entropy-Preserving Quantum Operations}\maketitle

\mbox{}\hrule\mbox{}\\~\\
\textbf{Abstract.} For a given quantum state $\rho$ and two quantum
operations $\Phi$ and $\Psi$, the information encoded in the quantum
state $\rho$ is quantified by its von Neumann entropy $\S(\rho)$. By
the famous Choi-Jamio{\l}kowski isomorphism, the quantum operation
$\Phi$ can be transformed into a bipartite state, the von Neumann
entropy $\S^{\mathrm{map}}(\Phi)$ of the bipartite state describes
the decoherence induced by $\Phi$. In this Letter, we characterize
not only the pairs $(\Phi, \rho)$ which satisfy
$\S(\Phi(\rho))=\S(\rho)$, but also the pairs $(\Phi, \Psi)$ which
satisfy $\S^{\mathrm{map}}(\Phi\circ\Psi) =
\S^{\mathrm{map}}(\Psi)$.
\mbox{}\hrule\mbox{}

\section{Introduction}

Von Neumann entropy and relative entropy are powerful tools in
quantum information theory. Both quantities have a monotonicity
property under a certain class of quantum operations and the
condition of equality is an interesting and important subject. For
example, an extremely important result was obtained in the condition
of equality for the celebrated strong subadditivity inequality for
the von Neumann entropy \cite{Hayden}. On the other hand, ones use
the von Neumann entropy of quantum states to  quantify the
information encoded in the quantum states, and the map entropy of
quantum operations to describe the decoherence induced by the
quantum operations \cite{Roga}. In this Letter, we study von Neumann
entropy-preserving quantum dynamical processes. In order to present
our results, we need the following notation.

For a given $N$-dimensional quantum-mechanical system which is
represented by an $N$-dimensional complex Hilbert space $\cH$, a
state $\rho$ on $\cH$ is a positive semi-definite operator of trace
one. If $\rho = \sum_k\lambda_k\out{u_k}{u_k}$ is its spectral
decomposition, then the support $\mathrm{supp}(\rho)$ of $\rho$ is
defined by $\mathrm{supp}(\rho)= \mathrm{span}\set{\ket{u_k} :
\lambda_k>0}$, and the generalized inverse $\rho^{-1}$ of $\rho$ is
defined by $\rho^{-1} =
\sum_{k:\lambda_k>0}\lambda^{-1}_k\out{u_k}{u_k}$. The \emph{von
Neumann entropy} $\S(\rho)$ of $\rho$ is defined by $\S(\rho) = -
\tr{\rho\log_{2}\rho}$ which quantifies information encoded in the
quantum state $\rho$. If $\sigma$ is also a quantum state on $\cH$,
then the \emph{relative entropy} between $\rho$ and $\sigma$ is
defined by $\S(\rho||\sigma) = \tr{\rho(\log_{2}\rho -
\log_{2}\sigma)}$ if $\mathrm{supp}(\rho) \subseteq
\mathrm{supp}(\sigma)$; $\S(\rho||\sigma) =  +\infty$, otherwise
\cite{Nielsen,Ohya}.

Let $\lin{\cH}$ be the set of all linear operators on $\cH$. If $X,
Y \in \lin{\cH}$, then $\inner{X}{Y} = \tr{X^{\dagger}Y}$ defines
the \emph{Hilbert-Schmidt inner product} on $\lin{\cH}$. Let
$\T(\cH)$ denote the set of all linear super-operators from
$\lin{\cH}$ to itself. $\Lambda\in \T(\cH)$ is said to be a
\emph{completely positive super-operator} if for each $k \in
\Natural$, $\Lambda\ot \idsup{\rM_{k}(\Complex)}: \lin{\cH} \ot
\rM_{k}(\Complex) \to \lin{\cH}\ot \rM_{k}(\Complex)$ is positive,
where $\rM_{k}(\Complex)$ is the set of all $k\times k$ complex
matrices. It follows from the famous theorem of Choi \cite{Choi}
that every completely positive super-operator $\Lambda$ has a Kraus
representation $\Lambda = \sum_{j}\mathrm{Ad}_{M_{j}}$, that is, for
every $X\in \lin{\cH}$,  $\Lambda (X) =
\sum_{j=1}^nM_jXM_j^\dagger$, where $\set{M_j}_{j=1}^n \subseteq
\lin{\cH}$, $M_j^\dagger$ is the adjoint operator of $M_j$. It is
clear that for the super-operator $\Lambda$, there is its
\emph{adjoint super-operator} $\Lambda^{\dagger}\in\T(\cH)$ such
that for $A, B\in\lin{\cH}$, $\inner{\Lambda(A)}{B} =
\inner{A}{\Lambda^{\dagger}(B)}$. Moreover, $\Lambda$ is a
completely positive super-operator if and only if $\Lambda^\dagger$
is also a completely positive super-operator.

A so-called \emph{quantum operation} is just a trace non-increasing
completely positive super-operator $\Phi$. If $\Phi$ is trace-preserving,
then it is called \emph{stochastic}; if $\Phi$ is stochastic and
unit-preserving, then it is called \emph{bi-stochastic}.

Let $\set{\ket{i} : i = 1,\ldots,N}$ be an orthonormal basis for
$\cH$. Then $\ket{\Omega} = \sum_{i=1}^N\ket{ii} \in \cH\ot\cH$ is a
maximal entangled state on $\cH\ot\cH$. The famous
\emph{Choi-Jamio{\l}kowski isomorphism} $J:
\T(\cH)\longrightarrow\lin{\cH\ot\cH}$ transforms every
$\Theta\in\T(\cH)$ into an operator $\jam{\Theta} =
(\Theta\ot\idsup{\lin{\cH}})(\out{\Omega}{\Omega})$ of
$\lin{\cH\ot\cH}$ \cite{Choi}. If $\Theta$ is completely positive,
then $\jam{\Theta}$ is a positive semi-define operator, in
particular, if $\Theta$ is stochastic, then $\frac1N\jam{\Theta}$ is
a bipartite state on $\cH\ot\cH$. For every stochastic quantum
operation $\Phi$, we refer to the von Neumann entropy
$\S(\frac1N\jam{\Phi})$ of the bipartite state $\frac1N\jam{\Phi}$
as the map entropy of $\Phi$, and denote it by
$\S^{\mathrm{map}}(\Phi)$, it describes the decoherence induced by
the quantum operation $\Phi$ \cite{Roga}.

In this Letter, for quantum states $\rho$ and quantum operations
$\Phi$ and $\Psi$, we characterize both the pairs $(\Phi,\rho)$
which satisfy $\S(\Phi(\rho))=\S(\rho)$ and the pairs $(\Phi,\Psi)$
which satisfy $\S^{\mathrm{map}}(\Phi\circ\Psi) =
\S^{\mathrm{map}}(\Psi)$.


Firstly, we need the following lemmas. For our purpose, Theorem 5.1
in \cite{Hiai} is modified into the following form:

\begin{lem}\label{Hiai}(\cite{Hiai})
Let $\rho$ and $\sigma$ be two quantum states on $\cH$, $\Phi$ be a
stochastic quantum operation. If $\mathrm{supp}(\rho) \subseteq
\mathrm{supp}(\sigma)$, then
$$\S(\Phi(\rho)||\Phi(\sigma)) = \S(\rho||\sigma)$$ if and only if $\Phi^\dagger_\sigma \circ \Phi(\rho) = \rho,$
where $\Phi^\dagger_\sigma = \mathrm{Ad}_{\sigma^{1/2}} \circ
\Phi^\dagger \circ \mathrm{Ad}_{\Phi(\sigma)^{-1/2}}$.
\end{lem}

\begin{lem}\label{Uhlm}(\cite{Lindblad})
Let $\rho$ and $\sigma$ be two quantum states on $\cH$, $\Phi$ be a
stochastic quantum operation. Then $$\S(\Phi(\rho)||\Phi(\sigma))
\leqslant \S(\rho||\sigma).$$
\end{lem}

\section{The Main Results and Proofs}

In general, we do not know whether a quantum operation will increase
or decrease the von Neumann entropy. However, if it is
bi-stochastic, then it does not decrease the entropy. That is, if
$\Phi$ is a bi-stochastic quantum operation, then $\S(\Phi(\rho))
\geqslant \S(\rho)$ for any quantum state $\rho$. This follows
immediately from Lemma~\ref{Uhlm} by letting $\sigma = \frac 1N
\idsup{\cH}$.

Our main results are:

\begin{thrm}\label{th:mainresult}
Let $\rho$ be a quantum state on $\cH$, and $\Phi$ be a
bi-stochastic quantum operation. Then the following statements are
equivalent:
\begin{enumerate}[(i)]
\item $\S(\Phi(\rho)) = \S(\rho)$.

\item $\Phi^\dagger \circ \Phi (\rho) = \rho$.

\item The space $\cH$ can be decomposed into:
$$\cH = \bigoplus_{k=1}^{K} \cH^L_k \ot \cH^R_k.$$
The quantum state $\rho$ can be decomposed into:
$$\rho=\bigoplus_{k=1}^{K} p_k \rho^L_k \ot \frac{1}{d^R_k} \idsup{\cH^R_k},$$
where $p = [p_1,\ldots,p_K]^\trans$ is a probability vector, that
is, all components of $p$ are non-negative and their sum is one, $d^R_k = \dim\cH^R_k$, $\rho^L_k$ is a quantum state on $\cH^L_k$, $k = 1,2,\ldots,K$. \\~\\
The bi-stochastic quantum operation $\Phi$ can be decomposed into:
$$\Phi = \bigoplus_{k=1}^{K} \Phi_k = \bigoplus_{k=1}^{K}\mathrm{Ad}_{U_k} \ot \Phi^R_k,$$
where $\Phi_k$ is the restriction of $\Phi$ to $\lin{\cH^L_k \ot
\cH^R_k}$, $\Phi_k = \mathrm{Ad}_{U_k} \ot \Phi^R_k$, $U_k \in
\lin{\cH^L_k}$ is a unitary operator and $\Phi^R_k \in \T(\cH^R_k)$
is a bi-stochastic and completely positive super-operator, $k=1,
\ldots, K$.
\end{enumerate}
\end{thrm}

\begin{proof}
$(i)\Longleftrightarrow(ii).$ Let $\sigma = \frac1N \idsup{\cH}$ in
Lemma \ref{Hiai}. Then $\mathrm{supp}(\rho) \subseteq
\mathrm{supp}(\sigma)$ and $\Phi^\dagger_\sigma = \Phi^\dagger$.
Note that $\Phi$ satisfies the conditions in Lemma~\ref{Hiai} and
$\Phi\Pa{\frac1N\idsup{\cH}} = \frac1N\idsup{\cH}$, so
$$
\S\Pa{\Phi(\rho)\left|\right|\frac1N\idsup{\cH}} =
\S\Pa{\rho\left|\right|\frac1N\idsup{\cH}}
$$
if and only if
$\Phi^\dagger \circ
\Phi(\rho)= \rho$, which implies that $\S(\Phi(\rho)) = \S(\rho)$ if
and only if $\Phi^\dagger\circ \Phi(\rho) = \rho$.

$(ii)\Longrightarrow(iii).$ Let $\sigma = \frac1N \idsup{\cH}$
again. Since $\Phi$ is bi-stochastic, we have $\mathrm{supp}(\sigma)
= \mathrm{supp}(\Phi(\sigma)) = \cH$ and $\Phi^\dagger_\sigma =
\Phi^\dagger$. Denote $$ \mathbf{Fix}(\Phi^\dagger\circ\Phi)=
\set{X\in\lin{\cH}: \Phi^\dagger\circ\Phi(X) = X}.
$$
It follows from Lemma 3.11 in
\cite{Hiai} that the space $\cH$ has a decomposition
$$
\cH =
\bigoplus_{k=1}^K \cH^L_k \ot \cH^R_k,
$$
such that
$$
\mathbf{Fix}(\Phi^\dagger\circ\Phi) =
\mathbf{Fix}(\Phi^\dagger_\sigma\circ\Phi) = \bigoplus_{k=1}^K
\lin{\cH^L_k}^+\ot\omega^R_k,
$$
and
$$
\Phi(X^L_k \ot \omega^R_k) = U_k X^L_k U^\dagger_k \ot
\tilde{\omega}^R_k,\, X^L_k \in \lin{\cH^L_k},
$$
where $U_k \in\lin{\cH^L_k}$ is a unitary operator, $\omega^R_k$ and
$\tilde{\omega}^R_k$ are two invertible quantum states on $\cH^R_k$,
$\lin{\cH^L_k}^+$ is the set of all positive semi-definite operators
in $\lin{\cH^L_k}$, $k = 1,\ldots,K$.

Note that $\Phi$ is a bi-stochastic quantum operation, so there is a
bi-stochastic quantum operation $\Phi^R_k$ on $\lin{\cH^R_k}$ such
that $\Phi^R_k( \omega^R_k)=\tilde{\omega}^R_k$. It follows from
$\idsup{\cH}\in \mathbf{Fix}(\Phi^\dagger\circ\Phi)$ that
$\omega^R_k = \tilde{\omega}^R_k = \frac{1}{d^R_k}\idsup{\cH^R_k}$,
where $d^R_k = \dim\cH^R_k$. The quantum state $\rho \in
\mathbf{Fix}(\Phi^\dagger\circ\Phi)$ implies that
$$\rho=\bigoplus_{k=1}^{K} p_k \rho^L_k \ot \frac{1}{d^R_k} \idsup{\cH^R_k},$$
where $p = [p_1,\ldots,p_K]^\trans$ is a probability vector,
$\rho^L_k$ is a quantum state on $\cH^L_k$, $k = 1,2,\ldots,K$, and
the decomposition of $\Phi$ is obtained immediately.

$(iii)\Longrightarrow(ii).$ Note that $\Phi^\dagger\circ\Phi =
\bigoplus_{k=1}^K \idsup{\lin{\cH^L_k}} \ot
(\Phi^R_k)^\dagger\circ\Phi^R_k$ and $\Phi^R_k$ is bi-stochastic, we
have $\Phi^R_k(\idsup{\cH^R_k}) = \idsup{\cH^R_k}$ and
$(\Phi^R_k)^\dagger(\idsup{\cH^R_k}) = \idsup{\cH^R_k}$. Thus,
$\Phi^\dagger\circ\Phi(\rho) = \rho$.
\end{proof}

\begin{remark}
We remark here that the completely positivity of $\Phi$ in
Theorem~\ref{th:mainresult} can be relaxed to be 2-positivity since
Lemmas~\ref{Hiai} and \ref{Uhlm} hold in that case.
\end{remark}

\begin{thrm}
Let $\Phi$ and $\Psi$ be two quantum operations on $\cH$, $\Phi$ be
bi-stochastic and $\Psi$ be stochastic. Then
$\S^{\mathrm{map}}(\Phi\circ\Psi) = \S^{\mathrm{map}}(\Psi)$ if and
only if $\Phi^\dagger\circ\Phi\circ\Psi = \Psi$.
\end{thrm}

\begin{proof}
It follows from $\S^{\mathrm{map}}(\Phi\circ\Psi) =
\S^{\mathrm{map}}(\Psi)$ that
$\S(\Phi\ot\idsup{\lin{\cH}}(\frac1N\jam{\Psi})) =
\S(\frac1N\jam{\Psi})$. Then, by Theorem \ref{th:mainresult},
$\S(\Phi\ot\idsup{\lin{\cH}}(\frac1N\jam{\Psi})) =
\S(\frac1N\jam{\Psi})$ holds if and only if $\Phi^\dagger\circ\Phi
\ot \idsup{\lin{\cH}}(\jam{\Psi}) = \jam{\Psi}$, that is,
$$\Phi^\dagger\circ\Phi\circ\Psi \ot \idsup{\lin{\cH}}(\out{\Omega}{\Omega}) = \Psi\ot\idsup{\lin{\cH}}(\out{\Omega}{\Omega}).$$
Note that $\jam{\Psi} = \sum_{i,j=1}^N\Psi(\out{i}{j}) \ot
\out{i}{j}$, so we have
$$\sum_{i,j=1}^N\Phi^\dagger\circ\Phi\circ\Psi(\out{i}{j}) \ot \out{i}{j} = \sum_{i,j=1}^N\Psi(\out{i}{j}) \ot \out{i}{j},$$
which implies that $\Phi^\dagger\circ\Phi\circ\Psi(\out{i}{j}) =
\Psi(\out{i}{j})$ for all $i,j$, this shows that
$\Phi^\dagger\circ\Phi\circ\Psi = \Psi$.
\end{proof}

\section{An Application}

If $p = [p_{1},\ldots,p_{N}]^{\trans} \in \Real^{N}$ and $q =
[q_{1},\ldots,q_{N}]^{\trans} \in \Real^{N}$ are two probability
vectors, the \emph{Shannon entropy} of $p$ and the \emph{relative
entropy} of $p$ and $q$ are defined by $\H(p) = - \sum_{i=1}^{N}
p_{i}\log_{2}p_{i}$ and $\H(p||q) = \sum_i p_i(\log_{2}p_i -
\log_{2}q_i)$, respectively, where $0\log_2 0=0$
\cite{Ohya,Shannon}).

Let $B = [b_{ij}]$ be an $N \times N$ matrix, if $b_{ij} \geqslant
0$, and $\sum_{i=1}^{N} b_{ij} = 1$ for every $j=1, \ldots, N$, then
$B = [b_{ij}]$ is said to be \emph{stochastic}; if $B = [b_{ij}]$ is
stochastic and  $\sum_{j=1}^{N} b_{ij} = 1$ for every $i=1, \ldots,
N$, then $B$ is said to be \emph{bi-stochastic}. Let $B$ be a
bi-stochastic $N\times N$ matrix, $p$ be an $N$-dimensional
probability vector. Then $Bp$ is also an $N$-dimensional probability
vector. In \cite{Alan}, A. Poritz and J. Poritz showed that $B$
preserves the Shannon entropy of $p$, that is, $\H(Bp)=\H(p)$, if and
only if $B^{\trans} B p = p$.

Now, we apply Theorem~\ref{th:mainresult} to prove the conclusion again. Firstly, we
need the following notion:

Let $\Phi$ be a stochastic quantum operation and
$\Phi=\sum_{\mu}\mathrm{Ad}_{M_{\mu}}$ be its Kraus representation.
Define the \emph{Kraus matrix} $B(\Phi)$ of $\Phi$ by
$B(\Phi)=\sum_{\mu}M_{\mu}\bullet \overline{M_{\mu}}$, where
$\bullet$ denotes the Shur product of matrices, that is, the
entrywise product of two matrices, and $\overline{M_{\mu}}$ is the
complex conjugate of $M_{\mu}$. It is easy to show that $B(\Phi)$ is
a stochastic matrix if $\Phi$ is a stochastic quantum operation on
$\cH$, and $B(\Phi)$ is a bi-stochastic matrix if $\Phi$ is a
bi-stochastic quantum operation on $\cH$. Moreover, $B(\Phi^\dagger)
= B(\Phi)^\trans$ \cite{Bengtsson}.

\begin{lem}\label{lem:bridge} Let $\set{\ket{i}: i = 1,\ldots, N}$ be an orthonormal basis
for $\cH$, $\Phi$ be a bi-stochastic quantum operation, $\rho$ be a
quantum state on $\cH$, $\sigma = \Phi(\rho)$, $p_j =
\Innerm{j}{\rho}{j}, q_j = \Innerm{j}{\sigma}{j}, j = 1, \ldots, N$,
$p = [p_1, \ldots, p_N]^\trans$, $q = [q_1, \ldots, q_N]^\trans$,
where all $p_{j}$ are the eigenvalues of $\rho$. Then $p$ and $q$
are two $N$-dimensional probability vectors and $q=B(\Phi)p$.
Conversely, if $p,q$ are two $N$-dimensional probability vectors and
$T$ is a bi-stochastic $N \times N$ matrix such that $q = Tp$, then
there exists a bi-stochastic quantum operation $\Phi$ on $\cH$ such
that $T = B(\Phi)$ and $\sigma = \Phi(\rho)$.
\end{lem}

\begin{proof} If $\Phi = \sum_\mu
\mathrm{Ad}_{M_\mu}$ is a Kraus representation of $\Phi$, then
\begin{eqnarray*}
q_i & = & \Innerm{i}{\Phi(\rho)}{i} = \sum_\mu \Innerm{i}{M_\mu \rho M^\dagger_\mu}{i} =  \sum_\mu \sum_{j=1}^N \Innerm{i}{M_\mu}{j} \Innerm{j}{\rho}{j} \Innerm{j}{M^\dagger_\mu}{i}\\
& = & \sum_{j=1}^N \sum_\mu
\Innerm{i}{M_\mu}{j}\overline{\Innerm{i}{M_\mu}{j}} p_j =
\sum_{j=1}^N \Innerm{i}{B(\Phi)}{j}p_j.
\end{eqnarray*}
That is $q=B(\Phi)p$. Conversely, let $T^\trans = [t_{ij}]$. Then we
can construct a stochastic quantum operation $\Phi_T$ on $\cH$ such
that $T = B(\Phi_T)$. In fact, $\Phi_T(\out{i}{i}) = \sum_j
t_{ij}\out{j}{j}$ defines a stochastic quantum operation, and for $\rho = \sum_i p_i\out{i}{i}, \sigma = \sum_j
(Tp)_j\out{j}{j}$, they satisfy
that
\begin{eqnarray*}
\Phi_T(\rho) & = & \sum_i p_i\Phi_T(\out{i}{i}) = \sum_i p_i\sum_j t_{ij}\out{j}{j} \\
& = & \sum_j (\sum_i p_i t_{ij})\out{j}{j} = \sum_j (Tp)_j\out{j}{j}
= \sigma.
\end{eqnarray*}
\end{proof}

\begin{cor}
Let $p$ be an $N$-dimensional probability vector and $B$ be an $N
\times N$ bi-stochastic matrix. Then  $\H(Bp) = \H(p)$ if and only
if $B^{\trans} B p = p$.
\end{cor}

\begin{proof}$\Longleftarrow$ is clear.

$\Longrightarrow$. Let $p = [p_1,\ldots,p_N]^\trans, q = Bp =
[q_1,\ldots,q_N]^\trans$, $\set{\ket{i}: i = 1,\ldots, N}$ be an
orthonormal basis for $\cH$. Define
$$\rho = \sum_{j=1}^{N}p_j\out{j}{j},\quad \sigma = \sum_{j=1}^{N}q_j\out{j}{j}.$$
Since $B$ is a bi-stochastic matrix, it follows from Lemma~\ref{lem:bridge} that there exists a bi-stochastic quantum operation
on $\cH$ such that $\sigma = \Phi(\rho)$, $B = B(\Phi)$. That
$\S(\Phi(\rho)) = \H(Bp)$ and $\S(\rho) = \H(p)$ is clear. Note that
$\H(Bp) = \H(p)$ implies that $\S(\Phi(\rho)) = \S(\rho)$. By
Theorem~\ref{th:mainresult} that $\S(\Phi(\rho)) = \S(\rho)$ if and
only if $\Phi^\dagger \circ \Phi(\rho) = \rho$, combining this fact
with $\sigma = \Phi(\rho)$, we have $\Phi^\dagger(\sigma) = \rho$.
It follows from Lemma~\ref{lem:bridge} again that $q = B(\Phi)p$, $p
= B(\Phi^\dagger)q = B(\Phi)^\trans q$. This shows that $p =
B(\Phi)^\trans B(\Phi) p$, that is, $p = B^\trans Bp$ since $B =
B(\Phi)$.
\end{proof}

\subsection*{Acknowledgement} The authors wish to express their
thanks to the referees for their valuable comments and suggestions.
This project is supported by Natural Science Foundations of China
(11171301, 10771191 and 10471124) and Natural Science Foundation of
Zhejiang Province of China (Y6090105).



\end{document}